\RequirePackage[l2tabu,orthodox]{nag}
\documentclass
[11pt,letterpaper]
{article} 
\usepackage[dvipsnames]{xcolor}
\usepackage[notes=true,later=false,camera=true]{dtrt}
\usepackage[utf8]{inputenc}
\usepackage{etex}
\usepackage{ stmaryrd }

\usepackage{xspace,enumerate}
\usepackage[T1]{fontenc}
\usepackage[full]{textcomp}
\usepackage[american]{babel}
\usepackage{mathtools}
\usepackage{amsthm}
\usepackage{empheq}

\usepackage{hyperref}
\hypersetup{hyperindex=true,pdfpagemode=UseOutlines,bookmarksnumbered=true,bookmarksopen=true,bookmarksopenlevel=2,pdfstartview=FitH,pdfborder={0 0 1},linkbordercolor=dt@linkcolor,citebordercolor=dt@linkcolor,urlbordercolor=dt@linkcolor}
\usepackage[capitalise,nameinlink]{cleveref}

\crefname{lemma}{Lemma}{Lemmas}
\crefname{fact}{Fact}{Facts}
\crefname{theorem}{Theorem}{Theorems}
\crefname{mtheorem}{Theorem}{Theorems}
\crefname{corollary}{Corollary}{Corollaries}
\crefname{claim}{Claim}{Claims}
\crefname{example}{Example}{Examples}
\crefname{algorithm}{Algorithm}{Algorithms}
\crefname{problem}{Problem}{Problems}
\crefname{definition}{Definition}{Definitions}
\crefname{property}{Property}{Properties}
\usepackage{paralist}
\usepackage{turnstile}
\usepackage{mdframed}
\usepackage{tikz}
\usepackage{caption}
\DeclareCaptionType{Algorithm}
\usepackage{newfloat}
\newtheorem{theorem}{Theorem}[section]
\newtheorem{mtheorem}{Theorem}%[section]
\newtheorem*{theorem*}{Theorem}

\newtheorem{proposition}[theorem]{Proposition}
\newtheorem*{proposition*}{Proposition}
\newtheorem{lemma}[theorem]{Lemma}
\newtheorem*{lemma*}{Lemma}

\newtheorem*{conjecture*}{Conjecture}

\newtheorem*{fact*}{Fact}

\newtheorem*{hypothesis*}{Hypothesis}

\theoremstyle{definition}
\newtheorem{definition}[theorem]{Definition}
\newtheorem*{definition*}{Definition}

\Crefname{customthm}{Theorem}{Theorems}

\theoremstyle{remark}
\newtheorem{claim}[theorem]{Claim}
\newtheorem*{claim*}{Claim}
\newtheorem{remark}[theorem]{Remark}
\newtheorem*{remark*}{Remark}

\newtheorem*{observation*}{Observation}

\newcommand{\Bigast}{\mathop{\scalebox{1.2}{\raisebox{-0.3ex}{$\ast$}}}}%

\usepackage[
letterpaper,
top=1.2in,
bottom=1.2in,
left=1in,
right=1in]{geometry}
\usepackage{newpxtext} % T1, lining figures in math, osf in text
\usepackage{textcomp} % required for special glyphs
\usepackage[varg,bigdelims]{newpxmath}
\usepackage[scr=rsfso]{mathalfa}% \mathscr is fancier than \mathcal
\usepackage{bm} % load after all math to give access to bold math
\let\mathbb\varmathbb
\usepackage{microtype}

\usepackage[backend=bibtex8,style=alphabetic,maxnames=99,sorting=anyt, minalphanames=5,maxalphanames=6,maxbibnames=99]{biblatex}
  \bibliography{sparse-rip}
\allowdisplaybreaks
\newcommand{\FormatAuthor}[3]{
	\begin{tabular}{c}
		#1 \\ {\small\texttt{#2}} \\ {\small #3}
	\end{tabular}
}

\newcommand{\keywords}[1]{\bigskip\par\noindent{\footnotesize\textbf{Keywords\/}: #1}}

\DeclareMathOperator{\eperiod}{\,\,\text{.}}
\DeclareMathOperator{\ecomma}{\,\,\text{,}}
\newcommand{\R}{{\mathbb R}}
\newcommand{\N}{{\mathbb N}}
\newcommand{\norm}[1]{\left\| #1 \right\|}

\newcommand{\abs}[1]{\lvert #1 \rvert}

\newcommand{\eps}{\varepsilon}

\newcommand{\F}{{\mathbb F}}

\newcommand{\E}{{\mathbb E}}

\newcommand{\ip}[1]{\langle #1 \rangle}

\newcommand{\tr}{\mathrm{tr}}

\newcommand{\ceil}[1]{\lceil #1 \rceil}

\DeclareMathOperator{\rank}{rank}%
\newcommand{\inparen}[1]{\left(#1\right)}
\newcommand{\inset}[1]{\left\{#1\right\}}
\newcommand{\deffont}[1]{\textnormal{\textsf{#1}}} 
\DeclareMathOperator{\supp}{supp}%
\newcommand{\inabs}[1]{\left|#1\right|}

\newcommand{\wt}{\text{wt}}

\newcommand{\polylog}{\mathrm{polylog}}

\let\svthefootnote\thefootnote
\newcommand\blfootnote[1]{%
	\let\thefootnote\relax%
	\footnotetext{#1}%
	\let\thefootnote\svthefootnote%
}
\begin{document}
	%%%%%%%%%%%%%%%%%%%%%%%%%%%%%%%%%%%%%%%%%%%%%%%%%%%%%%%%%%%%%%%%%%%%%%%%%%%%%%%%
	%%%%%%%%%%%%%%%%%%%%%%%%%%%%%%%%%%%%%%%%%%%%%%%%%%%%%%%%%%%%%%%%%%%%%%%%%%%%%%%%
	%%%%%%%%%%%%%%%%%%%%%%%%%%%%%%%%%%%%%%%%%%%%%%%%%%%%%%%%%%%%%%%%%%%%%%%%%%%%%%%%

	%%%%%%%%%%%%%%%%%%%%%%%%%%%%%%%%%%%%%%%%%%%%%%%%%%%%%%%%%%%%%%%%%%%%%%%%%%%%%%%%
	\title{Sparsity and $\ell_p$-Restricted Isometry}
	
 	\author{
 		\begin{tabular}[h!]{ccc}
 			\FormatAuthor{Venkatesan Guruswami\thanks{Supported in part by NSF grants CCF-1908125 and CCF-2210823, and a Simons Investigator Award.}}{venkatg@berkeley.edu}{UC Berkeley}
 			\FormatAuthor{Peter Manohar\thanks{Supported in part by an ARCS Scholarship, NSF Graduate Research Fellowship (under grant numbers DGE1745016 and DGE2140739), and NSF CCF-1814603.}}{pmanohar@cs.cmu.edu}{Carnegie Mellon University}
 			\FormatAuthor{Jonathan Mosheiff\thanks{Supported in part by NSF CCF-1814603.}}{mosheiff@bgu.ac.il}{Ben-Gurion University} 
 		\end{tabular}
 	} 
	\date{}
	%%%%%%%%%%%%%%%%%%%%%%%%%%%%%%%%%%%%%%%%%%%%%%%%%%%%%%%%%%%%%%%%%%%%%%%%%%%%%%%%

	%%%%%%%%%%%%%%%%%%%%%%%%%%%%%%%%%%%%%%%%%%%%%%%%%%%%%%%%%%%%%%%%%%%%%%%%%%%%%%%%
	\maketitle\blfootnote{Any opinions, findings, and conclusions or recommendations expressed in this material are those of the author(s) and do not necessarily reflect the views of the National Science Foundation.}
	\thispagestyle{empty}
	%%%%%%%%%%%%%%%%%%%%%%%%%%%%%%%%%%%%%%%%%%%%%%%%%%%%%%%%%%%%%%%%%%%%%%%%%%%%%%%%

	%%%%%%%%%%%%%%%%%%%%%%%%%%%%%%%%%%%%%%%%%%%%%%%%%%%%%%%%%%%%%%%%%%%%%%%%%%%%%%%%
	%%%%%%%%%%%%%%%%%%%%%%%%%%%%%%%%%%%%%%%%%%%%%%%%%%%%%%%%%%%%%%%%%%%%%%%%%%%%%%%%
	%%%%%%%%%%%%%%%%%%%%%%%%%%%%%%%%%%%%%%%%%%%%%%%%%%%%%%%%%%%%%%%%%%%%%%%%%%%%%%%%
	\begin{abstract}
		A matrix $A$ is said to have the $\ell_p$-Restricted Isometry Property ($\ell_p$-RIP) if for all vectors $x$ of up to some sparsity $k$, $\norm{Ax}_p$ is roughly proportional to $\norm{x}_p$.  We study this property for $m \times n$ matrices of rank proportional to $n$ and $k = \Theta(n)$. 
    In this parameter regime, $\ell_p$-RIP matrices are closely connected to Euclidean sections, and are ``real analogs'' of testing matrices for locally testable codes.
		
		It is known that with high probability, random dense $m\times n$ matrices (e.g., with i.i.d.\ $\pm 1$ entries) are $\ell_2$-RIP with $k \approx m/\log n$, and sparse random matrices are $\ell_p$-RIP for $p \in [1,2)$ when $k, m = \Theta(n)$.
		However, when $m = \Theta(n)$, sparse random matrices are known to \emph{not} be $\ell_2$-RIP with high probability.
		
		\medskip
Against	this backdrop, we show that sparse matrices \emph{cannot be} $\ell_2$-RIP in our parameter regime. On the other hand, for $p \neq 2$, we show that every $\ell_p$-RIP matrix \emph{must} be sparse.
		Thus, sparsity is incompatible with $\ell_2$-RIP, but necessary for $\ell_p$-RIP for $p \neq 2$. 
		
  Under a suitable interpretation, our negative result about $\ell_2$-RIP gives an impossibility result for a certain continuous analog of ``$c^3$-LTCs''---locally testable codes of constant rate, constant distance and constant locality that were constructed in recent breakthroughs.
		
		\keywords{Restricted Isometry Property, Sparse Matrices, Spread Subspaces, Euclidean Sections, Compressed sensing, Locally Testable Codes}
	\end{abstract}
	\pagestyle{plain}
	%\newpage
	
	%%%%%%%%%%%%%%%%%%%%%%%%%%%%%%%%%%%%%%%%%%%%%%%%%%%%%%%%%%%%%%%%%%%%%%%%%%%%%%%%
	%%%%%%%%%%%%%%%%%%%%%%%%%%%%%%%%%%%%%%%%%%%%%%%%%%%%%%%%%%%%%%%%%%%%%%%%%%%%%%%%
	%%%%%%%%%%%%%%%%%%%%%%%%%%%%%%%%%%%%%%%%%%%%%%%%%%%%%%%%%%%%%%%%%%%%%%%%%%%%%%%%
	\newpage
	\section{Introduction}\label{sec:intro}

	A random (dense) $0.01n \times n$ matrix $A$ with independent $\pm 1$ entries acts on all $o(n)$-sparse\footnote{A vector $x$ is $k$-sparse if the number of nonzero entries in $x$ is at most $k$.} vectors approximately isometrically. That is, $\norm{Ax}_2 \approx \sqrt{0.01n}\norm{x}_2$ for all $o(n)$-sparse vectors $x$. Such a matrix $A$ is said to satisfy the Restricted Isometry Property (RIP) for the $\ell_2$-norm. More generally, one can extend the definition of RIP to any $\ell_p$-norm.
	\begin{definition}\label{def:RIP}
		Let $p \geq 1$. A matrix $A \in \R^{m \times n}$ is \deffont{$(k, D)$-$\ell_p$-RIP} 
		if there exists $K > 0$ such that for every $k$-sparse $x \in \R^n$, it holds that 
		\begin{equation*}
			K\norm{x}_p \leq \norm{Ax}_p \leq D \cdot K\norm{x}_p \enspace.
		\end{equation*}
	\end{definition} 
	The original rise to prominence of the RIP is due to connections to compressed sensing~\cite{Donoho} unearthed in several works~\cite{CandesT05,CandesT06,CandesRT06}, where it was referred to as the Uniform Uncertainty Principle (UUP). The $\ell_p$-RIP can be used to achieve the so-called ``$\ell_p$-$\ell_1$'' guarantee for the sparse recovery problem: if $A$ is $(k,1+\eps)$-$\ell_p$-RIP and $x$ is $\delta$-close in $\ell_1$-distance to a sparse vector $y$, then one can approximately recover $x$ given the noisy measurement $Ax + b$, where $A$ is the measurement matrix and $b$ is the noise vector. Formally, one has the guarantee that the recovered $\hat{x}$ satisfies $\norm{x - \hat{x}}_p \leq O(k^{-(1- \frac{1}{p})}\delta + \norm{b}_p)$ \cite{ZhuGR15}.
	
	The most well-studied case of $\ell_p$-RIP is for $p = 2$. When $p = 2$, the RIP is equivalent to saying that the eigenvalues of $A_S^{\top} A_S$ are roughly equal for any small set $S \subseteq [n]$, where $A_S$ denotes the restriction of $A$ to the columns in $S$. It is well-known that for $m \geq O(k \log (n/k))$, a (dense) random $m \times n$ matrix with independent $\pm 1$ entries, or more generally any distribution that has the JL property (e.g., a subgaussian distribution)~\cite{BaranuikDDW08}, is $(k, O(1))$-$\ell_2$-RIP.
	
	The RIP also has close connections to \emph{spread subspaces}, an analog of error-correcting codes over the reals, as well as to Euclidean sections. A simple proof shows that the kernel of an $(\Omega(n), O(1))$-$\ell_p$-RIP matrix is a subspace with the ``$\ell_p$-spread property'': any $x \in \ker(A)$ with $\|x\|_p=1$ is $\Omega(1)$-far in $\ell_p$-distance from all $o(n)$-sparse vectors \cite[Prop.\ 3.8]{GMM22}. For $p = 2$, this corresponds to $\ker(A)$ being a good Euclidean section of $\ell_1$, a notion that has been studied in classical works~\cite{FLM77,kashin,garnaev-gluskin}, and more recently in \cite{KT07,GLW08,GLR10,Karnin11,GMM22}.	
	
	RIP matrices can also be interpreted as continuous analogs for testing matrices of linear locally testable codes. The testing matrix of a code is simply a matrix $A$ where each row corresponds to one of the possible linear local tests performed by the tester. 
 The code equals $\ker(A)$ and thus has dimension $n - \rank(A)$.
 The number of queries made by the tester is the row sparsity of $A$, The probability that $x$ fails the test is thus proportional to the Hamming weight $\wt(Ax)$ of the ``syndrome" $Ax$. In particular, the testing matrix $A$ of a (strong) locally testable code of constant rate $\rho$, constant locality $q$ and constant distance $\delta$ (``$c^3$-LTC'') will have rank $(1- \rho)n$, row sparsity $q$, and $\text{wt}(Ax) \geq \beta \cdot \wt(x)$ for all $x$ with $\wt(x) \leq \delta n /2$, where $\beta \in (0,1)$ is a constant. 
	Such codes were recently constructed in two breakthrough works \cite{DELLM22,PK22}. 

 In the continuous analog of codes, 
 the testing matrix $A$ for a $c^3$-LTC is analogous to an $(\Omega(n),O(1))$-$\ell_p$-RIP matrix $A$ of rank $\alpha n$ (for a constant $\alpha \in (0,1)$) whose rows have constant sparsity. The ``code" is $\ker(A)$ and its ``distance" corresponds to the fact that $\ker(A)$ is well-spread and in particular has no $\delta n$-sparse vectors. A testing matrix  should then, in particular, satisfy the following:  for all nonzero $x$ which are $\delta n/2$-sparse, the ``syndrome" $Ax$ has sizeable norm, specifically $\|Ax\|_p \geq \beta \|x\|_p$. Thus the analogy between RIP matrices and LTC testing matrices is achieved by replacing the finite field $\F_2$ with $\R$, and by replacing the Hamming ``metric'' $\wt(\cdot)$ with the $\ell_p$-norm $\norm{\cdot}_p$.\footnote{One may observe that the analogy only requires the lower bound in the RIP. However, the upper bound in the RIP also holds, as it is a simple consequence of the constant row sparsity, as the entries of $A$ are bounded by a constant.}
	
	In this work we study the sparsity of $(k, D)$-$\ell_p$-RIP matrices $A\in \R^{m\times n}$ in the regime where $k \ge \Omega(n)$, $D = O(1)$, and $\rank(A) \ge \Omega(n)$. This parameter setting of $k$, $D$, and $\rank(A)$ is naturally induced by the aforementioned relation of the RIP to $\ell_p$-spread subspaces and to testing of codes, but is less common in the context of compressed sensing and the JL property, where one typically has $k = o(n)$ and $\rank(A) \leq m = o(n)$. Sparsity naturally arises from the connection to testing of codes, where the testing matrix must be row sparse, as well as the connection to Euclidean sections/$\ell_p$-spread subspaces, where the best known explicit constructions, for all $p \in [1,2]$, come from the kernels of sparse matrices \cite{GLW08,GLR10, Karnin11}. For $p \ne 2$, it is known that random \emph{sparse} $m \times n$ matrices are $\ell_p$-RIP \cite{ZhuGR15,GMM22}, and explicit constructions, for $p \in [1,2]$, are also sparse \cite{BerindeGI+08,GMM22}. However, \cite{GMM22} also show that such matrices are \emph{not} $\ell_2$-RIP in the aforementioned parameter regime, when $k \geq \Omega(n)$ and $D = O(1)$.
	
	In this work, we thus ask the following questions: are $\ell_2$-RIP matrices necessarily dense? And is sparsity inherent to $\ell_p$-RIP matrices? 
	
	%%%%%%%%%%%%%%%%%%%%%%%%%%%%%%%%%%%%%%%%%%%%%%%%%%%%%%%%%%%%%%%%%%%%%%%%%%%%%%%%
	%%%%%%%%%%%%%%%%%%%%%%%%%%%%%%%%%%%%%%%%%%%%%%%%%%%%%%%%%%%%%%%%%%%%%%%%%%%%%%%%
	\subsection{Our results}
	\label{sec:results}
	We prove two theorems about the sparsity of $\ell_p$-RIP matrices. For $p = 2$ we show that any $\ell_2$-RIP matrix must contain a large number of rows with superconstant density, and for $p < 2$ we show that the rows of $A$ must be ``analytically sparse on average''. Taken together, our results show that sparsity is incompatible for $\ell_2$-RIP, but is necessary for $\ell_p$-RIP.

    Our result for $\ell_2$-RIP is stated informally below.
 \begin{mtheorem}[Simplified \cref{thm:ell2}]\label{ithm:ell2}
	Let $m,n\in \N$. Let $A\in \R^{m\times n}$ be a matrix of rank $\alpha n$ where $0 < \alpha < 1$ is a constant. Suppose that   $A$ is $(k,D)$-$\ell_2$-RIP for some $\Omega(n) \le k \le n$ and $1\le D\le O(1)$. Then for any constant $\eta \in (0,1)$, if we let $T$ denote the matrix consisting of all rows in $A$ that are not $s$-sparse for some $s \ge \Omega\inparen{\eta \log n}$, then: \begin{inparaenum}[(1)]
 \item $T$ contains at least $\Omega(n^{1-\eta})$ rows, and
\item $\norm{T}_F^2 \geq \Omega\inparen{n^{-\eta}\cdot \norm{A}_F^2}$, where $\norm{\cdot}_F$ denotes the Frobenius norm.
 \end{inparaenum}
	\end{mtheorem}

 	\cref{ithm:ell2} is not the only sparsity lower bound for $\ell_2$-RIP matrices. A simple argument given in \cite{ZhuGR15} (based on an argument in \cite{Chandar10}) shows that a $(k,D)$-$\ell_2$-RIP matrix $A$ with column sparsity $t$ must either have $m > n/k$ or $t \geq k/D^2$. This implies that if $k$ is small, e.g.\ $O(\log n)$, then the matrix $A$ must either have superconstant column sparsity ($\Omega(\log n)$) or have many rows ($\Omega(n/\log n)$, far larger than the $\polylog(n)$ rows required for dense matrices).

  \cref{ithm:ell2} is incomparable to the argument of \cite{Chandar10, ZhuGR15}, as it applies in the different parameter regime of $k = \Theta(n)$, which is the regime of importance for Euclidean sections and testing of codes. Indeed, the argument of \cite{Chandar10, ZhuGR15} does not give \emph{any} bound on the sparsity in the regime of $k = \Theta(n)$, as it only implies that the number of rows $m$ is at least some constant. This limitation is inherent to the proof technique, and so the argument of \cite{Chandar10, ZhuGR15} does not extend to the $k = \Theta(n)$ regime. Moreover, \cref{ithm:ell2} still holds even when $m \geq n$ (assuming that $\rank(A) = \alpha n$ for some constant $\alpha \in (0,1)$), whereas the argument of \cite{Chandar10, ZhuGR15} does not give any bound on the sparsity when $m \geq n$ (even if $k = 1$).

\cref{ithm:ell2} rules out the existence of an $\ell_2$-RIP analog of $c^3$-LTCs in the sense discussed earlier. Indeed, the theorem implies that a parity check matrix $A$ with ``constant rate'' ($\dim \ker A \ge \Omega(n)$) and ``constant distance'' ($A$ is $\inparen{\Omega(n),O(1)}$-$\ell_2$-RIP) must have rows of Hamming weight $\Omega(\log n)$.

We now turn to our theorem about $\ell_p$-RIP matrices for $p < 2$. 
\begin{mtheorem}[Simplified \cref{thm:ellp}]\label{ithm:ellp}
		Fix a constant $p\in [1,2]$ and let $A \in \R^{m \times n}$ be a $(k, D)$-$\ell_p$-RIP matrix for some $\Omega(n) \le k \le n$ and $1\le D\le O(1)$. Let $A_{1\Bigast}, \dots, A_{m\Bigast }$ denote the rows of $A$. Then, 
		\begin{flalign*}
			\sum_{i = 1}^m  \norm{A_{i\Bigast }}_2^{p} = \Theta \inparen{\sum_{i = 1}^m \norm{A_{i\Bigast }}_p^p} \enspace.
		\end{flalign*}
	\end{mtheorem}
 We note that the ``$O$'' part of \cref{ithm:ellp} is trivial, as $\sum_{i = 1}^m  \norm{A_{i\Bigast }}_2^{p} \leq \sum_{i = 1}^m  \norm{A_{i\Bigast }}_p^{p}$ always holds since $\norm{x}_2 \leq \norm{x}_p$ for $p \in [1,2]$. Thus, the ``$\Omega$'' part is the nontrivial statement.

\cref{ithm:ellp} is an analytic statement about the norms of the rows of $A$.
To explain its implications for the sparsity of $A$, let us first consider the following simple case where every row of $A$ has exactly $s$ nonzero entries, each of magnitude $1$. Then, any row $A_{i\Bigast }$ satisfies $\norm{A_{i\Bigast }}_p = s^{1/p - 1/2} \norm{A_{i\Bigast }}_2$, which implies that $O(1) \sum_{i = 1}^m  \norm{A_{i\Bigast }}_2^{p} \geq \sum_{i = 1}^m \norm{A_{i\Bigast }}_p^p = s^{1 - p/2}\sum_{i = 1}^m  \norm{A_{i\Bigast }}_2^{p}$, where the first inequality is due to \cref{ithm:ellp}. So, \cref{ithm:ellp} implies that $s^{1 - p/2} \leq O(1)$, i.e., $s$ is constant for any constant $p < 2$.
 
 More generally, for $p \in [1,2)$ and any vector $x \in \R^n$, H\"{o}lder's inequality implies that $\norm{x}_2 \leq \norm{x}_p \leq n^{\frac{1}{p} - \frac{1}{2}}\norm{x}_2$, with the lower inequality achieved when $x = e_i$ is a standard basis vector, i.e., very sparse, and the upper inequality achieved when $x = 1^n$, i.e., very dense. Thus, the ratio $\frac{\norm{x}_p}{\norm{x}_2}$ can be viewed as a notion of analytic sparsity for the vector $x$. In particular, if $\norm{x}_p = \Theta(\norm{x}_2)$, then a constant fraction of the $\ell_p$-mass of $x$ must lie on a constant number of coordinates.\footnote{See \cref{sec:SparseCompressibleDistorted,prop:Compressible-Distorted} for a more detailed discussion and formal statement.}
 \cref{ithm:ellp} informally says that $\norm{A_{i\Bigast }}_p = \Theta(\norm{A_{i\Bigast }}_2)$ ``on average''. One can thus interpret \cref{ithm:ellp} to say that the rows of $A$ must have ``constant analytic sparsity on average'', where ``analytic sparsity'' is in the $\ell_p$ vs. $\ell_2$ sense stated above.

We note that because \cref{ithm:ell2} is a statement about the \emph{average} of the rows, rather than a worst case statement about all rows, we cannot rule out the existence of an $\Omega(\log n)$-sparse matrix $A$ that is simultaneously $(\Omega(n), O(1))$-$\ell_2$-RIP and $(\Omega(n), O(1))$-$\ell_p$-RIP for some $1 \leq p < 2$. Indeed, if \cref{ithm:ellp} instead implied that $\norm{A_{i \Bigast}}_p = \Theta(\norm{A_{i \Bigast}}_2)$ for \emph{all} rows $i \in [m]$, then the rows of the submatrix $T$ from \cref{ithm:ell2} would violate this conclusion, implying that $A$ cannot be both $\ell_2$ and $\ell_p$-RIP. However, because \cref{ithm:ellp} is only a statement that holds on average, our results do not prove this.

\subsection{Proof overview}\label{sec:strategy}
The proof of \cref{ithm:ell2} has two key ideas. First, we define a new property, the \emph{analytic restricted isometry property} (ARIP, \cref{def:ARIP}), that strengthens \cref{def:RIP} by requiring not only that $\norm{Ax}_2 \approx \norm{x}_2$ for every $k$-sparse $x$, but also for vectors $x$ that are \emph{close} to $k$-sparse. We then show that one can convert between ARIP and RIP with some small loss in parameters when $k = \Theta(n)$ (\cref{prop:RIP-ARIP}). Hence, we can execute the proof of \cref{ithm:ell2} assuming that $A$ is ARIP.

Now, given that we assume that $A$ is ARIP, it suffices to show that if $A$ is too sparse, then there is a vector $x$ that is close to a $k$-sparse vector and $\norm{Ax}_2 \ll \norm{x}_2$. Our second key idea (\cref{lem:MainNegativeARIPWithoutDistortion}) is to choose $x = e^{-t A^{\top} A} e_i$, where $e_i$ is a standard basis vector chosen so that $\norm{\Pi e_i}_2$ is large, where $\Pi$ is the orthogonal projection onto $\ker(A)$ and $t > 0$ is a parameter that we will pick carefully.  In the eigenbasis of $A^{\top}A$, the transformation ${e^{-t A^\top A}}$ significantly attenuates the magnitude of eigenvectors with large eigenvalues but preserves the eigenvectors with eigenvalue $0$, i.e., those in $\ker(A)$. Thus, as $t \to \infty$, we have $x \to \Pi e_i$, i.e., $x$ becomes the projection of $e_i$ onto $\ker(A)$.
So, as $t \to \infty$, we have $\norm{Ax}_2 \to 0$, but $\norm{x}_2$ will still be large because $e_i$ was chosen to have large projection onto $\ker(A)$. 

To violate the ARIP of $A$, we additionally need to have that $x$ is close to $k$-sparse, and if we take $t \to \infty$ it becomes hard to control this quantity. But, on the other extreme when $t = 0$, we have $x = e_i$, which is $1$-sparse. By carefully choosing $t$ to be an intermediate quantity, we can simultaneously ensure that $\norm{Ax}_2$ is small, $\norm{x}_2$ is large, and $x$ is close to sparse. We control the ``approximate sparsity'' of $x$ by bounding $\norm{x}_1$, and here we crucially use the sparsity of $A$.

The full proof of \cref{ithm:ell2} requires some additional steps, but \cref{lem:MainNegativeARIPWithoutDistortion} captures the core of the argument.

The proof of \cref{ithm:ellp} is simpler. We probe the matrix $A$ with a random standard coordinate vector $e_j$ and a random $k$-sparse vector $g$ with Gaussian entries. We show that $\norm{Ae_j}_p^p/\norm{e_j}_p^p = t$ behaves very differently in expectation compared to $\norm{Ag}_p^p/\norm{g}_p^p$, and this allows us to bound $\sum_{i = 1}^m \norm{A_{i \Bigast}}_{2}^p$ in terms of $\sum_{i = 1}^m \norm{A_{i \Bigast}}_{p}^p$.

     %%%%%%%%%%%%%%%%%%%%%%%%%%%%%%%%%%%%%%%%%%%%%%%%%%%%%%%%%%%%%%%%%%%%%%%%%%%%%%%%
    %%%%%%%%%%%%%%%%%%%%%%%%%%%%%%%%%%%%%%%%%%%%%%%%%%%%%%%%%%%%%%%%%%%%%%%%%%%%%%%%
	%%%%%%%%%%%%%%%%%%%%%%%%%%%%%%%%%%%%%%%%%%%%%%%%%%%%%%%%%%%%%%%%%%%%%%%%%%%%%%%%
 \section{Preliminaries}
 \label{sec:prelims}
     %%%%%%%%%%%%%%%%%%%%%%%%%%%%%%%%%%%%%%%%%%%%%%%%%%%%%%%%%%%%%%%%%%%%%%%%%%%%%%%%
	%%%%%%%%%%%%%%%%%%%%%%%%%%%%%%%%%%%%%%%%%%%%%%%%%%%%%%%%%%%%%%%%%%%%%%%%%%%%%%%%
 \subsection{Notation and conventions}
	The implicit factor in asymptotic notations is an absolute constant, unless stated otherwise.
	
    For a matrix $A\in \R^{m\times n}$, we denote the $i$-th row and the $j$-th column of $A$ by $A_{i\Bigast }$ and $A_{\Bigast j}$, respectively.
      
        Let $p \ge 1$. For a vector $x \in \R^n$, we let $\norm{x}_p$ denote the $\ell_p$-norm of $x$, i.e., $\norm{x}_p = (\sum_{i = 1}^n \abs{x_i}^p)^{1/p}$. For a matrix $A \in \R^{m \times n}$, we let $\norm{A}_{\ell_p \to \ell_p}$ denote the $\ell_p\to\ell_p$ operator norm of a matrix $A$, which is defined by
	\begin{equation*}
	\norm{A}_{\ell_p\to\ell_p} = \sup\inset {\norm {Ax}_p \mid x\in \R^n~\wedge~\norm x_p = 1}\eperiod
	\end{equation*}
	The \deffont{Frobenius norm} $\norm{A}_F$ is defined as $\sqrt{\sum_{i,j} \abs{A_{i,j}}^2}$. The Frobenius norm satisfies
 \begin{equation}\label{eq:FrobeniusTrace}
     \norm A_F^2 = \tr\inparen{AA^\top} = \tr\inparen{A^\top A}\eperiod
 \end{equation}

     %%%%%%%%%%%%%%%%%%%%%%%%%%%%%%%%%%%%%%%%%%%%%%%%%%%%%%%%%%%%%%%%%%%%%%%%%%%%%%%%
	%%%%%%%%%%%%%%%%%%%%%%%%%%%%%%%%%%%%%%%%%%%%%%%%%%%%%%%%%%%%%%%%%%%%%%%%%%%%%%%%
	\subsection{Sparse, Compressible and Distorted vectors}\label{sec:SparseCompressibleDistorted}
	\begin{definition}
		Let $x\in \R^n$ and let $1\le q\le p$. We define the \deffont{$\inparen{\ell_q,\ell_p}$-distortion} of $x$ as
		$$\Delta_{q,p}(x) = \frac{\norm x_p \cdot n^{1/q-1/p}}{\norm x_q}\eperiod$$
	\end{definition}
	H\"older's inequality yields
	\begin{equation}\label{eq:DistortionBounds}
		1 \le \inparen{\frac n{\inabs{\supp(x)}}}^{1/q-1/p}\le \Delta_{q,p}(x)\le n^{1/q-1/p}\eperiod
	\end{equation}	
	In particular, sparse vectors have large distortion.
 As we next discuss, a certain converse of this fact also holds.
	
	\begin{definition}
		Fix $p \ge 1$. Let $k\le n\in \N$ and $\eps > 0$. A vectors $x\in \R^n\setminus \inset{0}$ is \deffont{$(k,\eps)$-$\ell_p$-compressible} if there exists a $k$-sparse $y\in \R^n$ such that $\norm{x-y}_p\le \eps \norm{x}_p$. 
	\end{definition}
	\begin{remark} \label{rem:largestEntriesForCompression}
		Without loss of generality $y$ can be taken to be equal to $x$ in the $k$ entries of $x$ that have the largest absolute value, and $0$ everywhere else.
	\end{remark}
	
	\begin{proposition}[{\cite[Prop.\ 3.11]{GMM22}}]\label{prop:Compressible-Distorted}
		Fix $1\le q < p$, $k\le n\in \N$  and $x\in \R^n$. The following holds.
		\begin{enumerate}
			\item \label{item:comp->dist} Let $\eps > 0$. If $x$ is $(k,\eps)$-$\ell_p$-compressible then $\Delta_{q,p}(x) \ge \frac{1}{\inparen{\frac kn}^{1/q-1/p}+\eps}$.
			\item \label{item:dist->comp} The vector $x$ is $\inparen{k,\frac{\inparen{\frac nk}^{1/q}}{\Delta_{q,p}(x)}}$-$\ell_p$-compressible.
		\end{enumerate}
	\end{proposition}
	
	We can now state the full version of the main theorems and prove that they imply their specialized versions from \cref{sec:intro}.

	%%%%%%%%%%%%%%%%%%%%%%%%%%%%%%%%%%%%%%%%%%%%%%%%%%%%%%%%%%%%%%%%%%%%%%%%%%%%%%%%
	%%%%%%%%%%%%%%%%%%%%%%%%%%%%%%%%%%%%%%%%%%%%%%%%%%%%%%%%%%%%%%%%%%%%%%%%%%%%%%%%
	%%%%%%%%%%%%%%%%%%%%%%%%%%%%%%%%%%%%%%%%%%%%%%%%%%%%%%%%%%%%%%%%%%%%%%%%%%%%%%%%

	\section{$\ell_2$-RIP Matrices Cannot Be Sparse}\label{sec:ell2Proof}
	In this section, we prove \cref{thm:ell2}, stated below, which is the formal version of \cref{ithm:ell2}.
 \begin{mtheorem}\label{thm:ell2}
		Fix $m,n\in \N$. Let $A\in \R^{m\times n}$ be a matrix of rank $r = \alpha n$ for some $0 < \alpha < 1$. Suppose that $A$ is $(k,D)$-$\ell_2$-RIP for some $k\ge 1$ and $D\ge 1$. Fix $\eta > 0$ such that $\frac{C_1 D^8}{(1-\alpha)^4}\le n^\eta \le \frac{C_2k^3}{D^4n^2}$, where $C_1,C_2$ are universal positive constants. Then, there is a submatrix $T\in \R^{t\times n}$ of rows of $A$, for some $t \le m$, with the following properties.
            
		\begin{enumerate}
			\item \label[property]{property:TLargeEll2RIP} $\frac{\norm {T}_F^2}{\norm A_F^2} \ge n^{-\eta}$.
			\item \label[property]{property:TmanyRowsRIP} $t \ge \Omega\inparen{ \frac{k^3}{D^4n^{2+\eta}}}$. 
   			\item \label[property]{property:TDistortedRIP} Every row $x$ of $T$ satisfies 
                \begin{equation*}
   			    \Delta_{1,2}(x) \le O\inparen{\inparen{\frac{D^3 n}{\eta \sqrt{1-\alpha}\log n}}^{1/2}} \enspace.
   		\end{equation*}
            In particular, every row is $\sqrt{\frac{s}{n}}$-far from all $s$-sparse vectors, where $s = \Omega\inparen{\frac{\eta \sqrt{1-\alpha}\log n}{D^3}}$.
		\end{enumerate}
	\end{mtheorem}
 
	\subsection{The Analytic Restricted Isometry Property and its equivalence to RIP}
    To prove \cref{thm:ell2}, we introduce the \deffont{Analytic Restricted Isometry Property (ARIP)} --- a natural strengthening of the RIP where we require that $\norm{Ax}_2 \approx \norm{x}_2$ for not just all $k$-sparse $x$, but all $x$ that is ``analytically $k$-sparse'', i.e., $\Delta_{1,2}(x) \geq (n/k)^{1/2}$.\footnote{Note that if $x$ is $k$-sparse, then $\norm{x}_1 \leq \sqrt{k} \norm{x}_2$, which implies that $\Delta_{1,2}(x) \geq (n/k)^{1/2}$.}
    Using \cref{prop:Compressible-Distorted}, it is straightforward to show that the two notions are related, and are essentially equivalent in a suitable parameter regime. We then prove \cref{thm:ell2} for matrices that are ARIP, and use the equivalence to complete the proof.
 
    We now define the ARIP below. For simplicity, we will focus here on the $\ell_2$-norm, although the definition makes sense for any $\ell_p$-norm as well.
	\begin{definition}
 \label{def:ARIP}
		A matrix $A\in \R^{m\times n}$ is \deffont{$(k,D)$-$\ell_2$-ARIP} if there exists $K > 0$ such that for every $x\in \R^n$ with $\Delta_{1,2}(x) \ge \inparen{\frac nk}^{1/2}$ it holds that 
		$$K\norm x_2 \le \norm{Ax}_2\le D\cdot K\norm{ x}_2\eperiod$$
	\end{definition}

        By \cref{eq:DistortionBounds}, ARIP immediately implies RIP. As we show next, a reverse implication also holds at a certain cost to the parameters.
	
	\begin{proposition}[Equivalence between RIP and ARIP]\label{prop:RIP-ARIP}
		Fix $k>0$ and $D\ge 1$, and let $A\in \R^{m\times n}$. The following holds.
		\begin{enumerate}
			\item If $A$ is $(k,D)$-$\ell_2$-ARIP then it is also $(k,D)$-$\ell_2$-RIP.
			\item Let $D' > D$ and suppose that $A$ is $(k,D)$-$\ell_2$-RIP. Then, $A$ is $(k',D')$-$\ell_2$-ARIP for $k' = \frac{(D'-D)^2k^3}{(D'D+D'+D)^2n^2}$. \\ In particular, if $k \ge \Omega(n)$ and $D \le O(1)$ then $A$ is $(\Omega(n),2D)$-$\ell_2$-ARIP.
		\end{enumerate}
	\end{proposition}

	\begin{proof}
		The first claim follows immediately from \cref{eq:DistortionBounds}. We turn to proving the second claim.
		
		We assume without loss of generality that $A$ satisfies \cref{def:RIP} with $K=1$. Fix $k' > 0$. Let $x\in \R^n$ with $\norm x_2=1$ and suppose that $\Delta_{1,2}(x) \ge \inparen{\frac n{k'}}^{1/2}$. Our goal is to show that 
		\begin{equation}\label{eq:RIPtoARIPNeedToShow}
			1-(D+1)\cdot \frac {n\cdot k'^{1/2}}{k^{3/2}} \le \norm {Ax}_2 \le D\inparen{1+\frac {n\cdot k'^{1/2}}{k^{3/2}}}\eperiod
		\end{equation} 
		\cref{eq:RIPtoARIPNeedToShow} yields the claim, since, taking $k'$ as in the proposition statement, the ratio between the right-hand side and the left-hand side of \cref{eq:RIPtoARIPNeedToShow} becomes at most $D'$. This implies that $A$ is $(k',D')$-$\ell_2$-ARIP. We turn to proving \cref{eq:RIPtoARIPNeedToShow}.
		
		Suppose, without loss of generality, that the entries of $x$ are sorted in order of non-increasing absolute value. Write $x = \sum_{j=1}^{\ceil{n/k}} y^j$, where $y^j$ is the $k$-sparse vector defined by
		$$
		y^j_i =\begin{cases}
			x_i &\text{if }(j-1)k< i\le jk\\
			0 &\text{otherwise.}
		\end{cases}
		$$
		Denote $y' = x - y^1 = \sum_{j=2}^{\ceil {n/k}}y^j $. By \cref{prop:Compressible-Distorted}, $x$ is $\inparen{k,\frac{\sqrt{nk'}}{k}}$-$\ell_2$-compressible, so \cref{rem:largestEntriesForCompression} yields
		\begin{equation}\label{eq:y'NormBound}
			\norm {y'}_2 = \norm{x-y^1}_2 \le  \frac{\sqrt{nk'}}k\eperiod
		\end{equation}

		Now, by the triangle inequality,
		\begin{equation}\label{eq:RIPARIPAxBound}
			\norm{Ay^1}_2 - \norm{Ay'}_2 \le \norm{Ax}_2 \le \norm{Ay^1}_2 + \norm{Ay'}_2\eperiod
		\end{equation}
		By the RIP assumption,
		\begin{equation}\label{eq:RIPARIPAy1Bound}
			1-\norm{y'}_2=\norm {x}_2-\norm{y'}_2\le \norm{y^1}_2 \le \norm{Ay^1}_2 \le D\norm{y^1}_2\le D\norm x_2 = D\eperiod
		\end{equation}
		Also, by the RIP assumption and H\"older's inequality,
		\begin{equation}\label{eq:RIPARIPAy'Bound}
			\norm{Ay'}_2 \le \sum_{j=2}^{\ceil {n/k}}\norm{Ay^j}_2 \le D\sum_{j=2}^{\ceil {n/k}}\norm {y^j}_2 \le D\inparen{\frac nk}^{1/2} \norm {y'}_2\eperiod
		\end{equation}
		Together, \cref{eq:RIPARIPAxBound,eq:RIPARIPAy1Bound,eq:RIPARIPAy'Bound} yield
		$$1- \inparen{D+1}\inparen{\frac nk}^{1/2}\norm{y'}_2\le 1-\inparen{1+D\inparen{\frac nk}^{1/2}}\norm{y'}_2\le \norm{Ax}_2 \le D+D\inparen{\frac nk}^{1/2}\norm{y'}_2\eperiod$$
		\cref{eq:RIPtoARIPNeedToShow} follows from the above and \cref{eq:y'NormBound}.
	\end{proof}

        We shall now state \cref{thm:ell2ARIP} -- a version of \cref{thm:ell2} for ARIP matrices -- and immediately prove that the former implies the latter. The rest of this section will be devoted to proving \cref{thm:ell2ARIP}.
 	\begin{mtheorem}[\cref{thm:ell2} for ARIP matrices]\label{thm:ell2ARIP}
		Fix $m,n\in \N$. Let $A\in \R^{m\times n}$ be a matrix of rank $r = \alpha n$ for some $0 < \alpha < 1$. Suppose that $A$ is $(k,D)$-$\ell_2$-ARIP for some $k\ge 1$ and $D\ge 1$. Fix $\eta > 0$ such that $\frac{64D^8}{(1-\alpha)^4}\le n^\eta \le \frac{k}{D^2}$. Then, there is a submatrix $T\in \R^{t\times n}$ of rows of $A$, for some $t \le m$, with the following properties.
            
		\begin{enumerate}
			\item \label[property]{property:TLargeEll2ARIP} $\frac{\norm {T}_F^2}{\norm A_F^2} \ge n^{-\eta}$.
			\item \label[property]{property:TmanyRowsARIP} $t \ge \frac{k}{D^2n^\eta}$. 
   			\item \label[property]{property:TDistortedARIP} Every row $x$ of $T$ satisfies 
                \begin{equation}
                \label{eq-arip-distortion}
   			    \Delta_{1,2}(x) \le O\inparen{\inparen{\frac{D^3 n}{\eta \sqrt{1-\alpha}\log n}}^{1/2}} \enspace.
   		\end{equation}
            In particular, every row is $\sqrt{\frac{s}{n}}$-far from all $s$-sparse vectors, where $s = \Omega\inparen{\frac{\eta \sqrt{1-\alpha}\log n}{D^3}}$.
		\end{enumerate}
	\end{mtheorem}
        \begin{proof}[Proof of \cref{thm:ell2} given \cref{thm:ell2ARIP}]
            Let $A\in \R^{n\times m}$ be $(k,D)$-$\ell_2$-RIP. By \cref{prop:RIP-ARIP}, $A$ is $(k',D')$-$\ell_2$-ARIP for $D' = 2D$ and $k' = \frac{D^2k^3}{(2D^2+3D)^2n^2} \ge \Omega\inparen{\frac {k^3}{D^2n^2}}$. The conclusion of \cref{thm:ell2} follows by applying \cref{thm:ell2ARIP} to $A$.
        \end{proof}
	
	\subsection{Distortion bounds for $\ell_2$-ARIP matrices}	
    We next develop the necessary tools to prove \cref{thm:ell2ARIP}. The following lemma states several simple but useful facts about ARIP matrices.
    \begin{lemma}\label{lem:BasicBounds}
		Let $A\in \R^{m\times n}$ be $(k,D)$-$\ell_2$-ARIP. Assume that $\norm A_F = \sqrt n$. The following then holds.
		\begin{enumerate}
			\item \label{enum:UnitBound} There exist some $i_{\mathrm{less}},i_{\mathrm{more}}\in \{1,\dots,n\}$ such that $\norm{Ae_{i_{\mathrm{less}}}}_2 \le 1 \le \norm{Ae_{i_{\mathrm{more}}}}_2$.
			\item \label{enum:KernelBound} Let $\Pi \in \R^{n \times n}$ be the projection matrix for the orthogonal projection onto $\ker(A)$. Then, there exists some $i_{\ker}\in \{1,\dots,n\}$ such that $\norm{\Pi e_{i_{\ker}}}_2 \ge \sqrt{1-\frac{\rank(A)}n}$.
			\item \label{enum:OneSideARIP} For every $x\in \R^n$ with $\Delta_{1,2}(x) \ge \inparen{\frac nk}^{1/2}$, it holds that $\frac{\norm x_2}{D }\le\norm{Ax}_2 \le D\norm{x}_2$.
			\item \label{enum:RowBound} Each row of $A$ is of $\ell_2$-norm at most $D\inparen{\frac nk}^{1/2}$.
			\item \label{enum:ColumnBound} Each column of $A$ is of $\ell_2$-norm at most $D$.
		\end{enumerate}
	\end{lemma}
	\begin{proof}
			We begin with \cref{enum:UnitBound,enum:KernelBound}. Let $i'$ be uniformly sampled from $\{1,\dots, n\}$. Then, 
			$$\E{\norm{Ae_{i'}}_2^2} = \frac 1n \sum_{i=1}^n \norm{Ae_{i}}_2^2 = \frac 1n \sum_{i=1}^n\sum_{j=1}^m A_{i,j}^2 = \frac{\norm A_F^2}n = 1\eperiod$$
			which yields \cref{enum:UnitBound} since there must exist choices of $i'$ for which $\norm{Ae_i'}_2^2$ is at most (resp. at least) its expectation. \cref{enum:KernelBound} follows similarly from
			$$\E{\norm{\Pi e_{i'}}_2^2} = \frac{\norm \Pi_F^2}n = \frac{\dim( \ker(A))}n = 1 - \frac{\rank(A)}n\eperiod$$
			
			For \cref{enum:OneSideARIP} recall that according to the ARIP assumption, there is some $K > 0$ such that $$K\norm x_2 \le \norm{Ax}_2\le D\cdot K\norm{x}_2$$ for all $x\in \R^n$ with $\Delta_{1,2}(x)\ge \inparen{\frac nk}^{1/2}$. It thus suffices to show that $K\le 1\le DK$. We use \cref{enum:UnitBound}, substituting $e_{\mathrm{less}}$, for $x$. This yields 
			$$K = K\norm{e_{\mathrm{less}}}_2 \le \norm{Ae_{\mathrm{less}}}_2 \le 1\eperiod$$
			Similarly,
			$$KD = KD\norm{e_{\mathrm{more}}}_2 \ge \norm{Ae_{\mathrm{more}}}_2 \ge 1\ecomma$$
			proving the claim.
			
			We turn to \cref{enum:RowBound}. Suppose towards contradiction that $A$ has a row $A_{i\Bigast}$ with $\norm {A_{i\Bigast}}_2 > D\cdot \sqrt{\frac {n}k}$. Let $J$ be a set of $k$ coordinates such that the respective entries of $A_{i\Bigast}$ are the $k$ largest in absolute value. Define $b\in \R^n$ by 
			$$
			b_j = 
			\begin{cases}
				A_{i,j} &\text{if }j\in J\\
				0&\text{otherwise.}
			\end{cases}
			$$
			In particular, $b$ is $k$-sparse and so $\Delta_{1,2}(b)\ge \inparen{\frac nk}^{1/2}$ due to \cref{eq:DistortionBounds}. Hence, $$\frac{\norm{Ab}_2}{\norm{b}_2} \ge \frac{(Ab)_i}{\norm b_2} = \frac{\ip{b,b}}{\norm b_2} = \norm b_2 = \inparen{\sum_{j\in J}A_{i,j}^2}^{1/2} \ge \inparen{\frac kn\sum_{j=1}^n A_{i,j}^2}^{1/2} = \sqrt{\frac kn}\cdot \norm{A_{i\Bigast}}_2 > D \ , $$
			which contradicts \cref{enum:OneSideARIP}.
						
			Finally, for \cref{enum:ColumnBound}, observe that \cref{enum:OneSideARIP} yields $\norm {A_{\Bigast i}}_2 = \norm{Ae_i}_2 \le D$ for all $1\le i\le n$. \qedhere
	\end{proof}

    With \cref{lem:BasicBounds}, we are now ready to prove our key technical lemma, \cref{lem:MainNegativeARIPWithoutDistortion}. This lemma contains the technical core of our argument. 
        \begin{lemma}\label{lem:MainNegativeARIPWithoutDistortion}
		Fix $m,n\in \N$. Let $A\in \R^{m\times n}$ be a matrix of rank $r = \alpha n$ for some $0 < \alpha < 1$. Suppose that $A$ is $(k,D)$-$\ell_2$-ARIP for some $k \ge 1$ and $D\ge 1$. Then, 
		\begin{equation}\label{eq:MainNegativeExact}
                \frac{\norm{A^\top A}_{\ell_1\to \ell_1}}{\norm A_F^2} \ge \frac{e\sqrt{1-\alpha}\ln\inparen{ k(1-\alpha)}}{Dn}\eperiod
		\end{equation}
		In particular, for $D\le O(1)$, $1-\alpha \ge \Omega(1)$ and $k \ge \Omega(n)$ we have \begin{equation}\label{eq:MainNegativeApprox}
			\frac{\norm{A^\top A}_{\ell_1\to \ell_1}}{\norm A_F^2}\ge \Omega\inparen{\frac {\log n}{n}}\eperiod
		\end{equation}
	\end{lemma}
    \begin{remark}\label{rem:MainLemmaForBoundedMatrix}
        To understand the statement of the lemma, it is helpful to consider the special case where $A$ is an $(\Omega(n), O(1))$-$\ell_2$-ARIP matrix with entries of magnitude $1$, $m = \Theta(n)$, and exactly $s$ (resp.\ $t$) nonzero entries per row (resp.\ column). For such a matrix $A$, each entry of $A^\top A$ is at most $st$, and $\norm{A}_F^2 = t n$. \cref{eq:MainNegativeApprox} then implies that $s \ge \Omega (\log n)$. In particular, $A$ cannot have constant row sparsity.
    \end{remark}
 \begin{proof}
        Let $B$ denote the positive semi-definite matrix $A^\top A\in \R^{n\times n}$. Write $\lambda_1,\dotsc, \lambda_n$ for the eigenvalues of $B$, and $v_1,\dotsc, v_n$ for the respective eigenvectors. Note that $B$ has exactly $r$ nonzero eigenvalues, so we may assume that $\lambda_1,\dots, \lambda_r$ are positive, while $\lambda_{r+1}=\dots=\lambda_n=0$. We further assume, without loss of generality, that $\norm A_F = \sqrt n$, and consequently, $\tr(B) = \sum_{i=1}^n \lambda_i = \norm A_F^2 = n$. 

        \cref{eq:MainNegativeApprox} clearly follows from \cref{eq:MainNegativeExact}, so it suffices to prove the latter. Our proof proceeds as outlined in \cref{sec:strategy}.

        Let $\Pi:\R^n\to \R^n$ denote the orthogonal projection map onto $\ker(A)$. By \cref{lem:BasicBounds}(\ref{enum:KernelBound}), there exists some $1\le i_{{\ker}} \le n$ such that $\norm{\Pi e_{i_{\ker}}}_2 \ge \sqrt{1-\alpha}$. We take $x = e^{-tB}e_{i_{\ker}}$ for some $t > 0$ to be determined later. Recall that $e^{-t B}$ is defined to be the matrix $\sum_{j=0}^\infty \frac{(-tB)^j}{j!}$. We claim that $x$ satisfies the following properties.
		\begin{enumerate}
			\item $\norm x_2 \ge \sqrt{1-\alpha}$ \label[property]{property:x_2Large}
			\item $\norm x_1 \le e^{t\norm B_{\ell_1\to \ell_1}}$ \label[property]{property:x_1Small}
			\item $\norm{Ax}_2 \le \frac 1{2te}$ \label[property]{property:Ax_2Small}
		\end{enumerate}
		Before proving these properties, we show that they yield \cref{eq:MainNegativeExact} and, consequently, the lemma. Take $t = \frac{\ln\inparen{\sqrt{k\inparen{1-\alpha}}}}{\norm B_{\ell_1\to \ell_1}}$. Then, $$\Delta_{1,2}(x) = \frac{\sqrt n\norm x_2}{\norm x_1} \ge \frac{\sqrt {n(1-\alpha)}}{e^{t\norm B_{\ell_1\to \ell_1}}} = \sqrt {\frac nk}\eperiod$$ 
		Therefore, by \cref{lem:BasicBounds}(\ref{enum:OneSideARIP}), 
		\begin{equation*}
			D \ge \frac{\norm {x}_2}{\norm {Ax}_2} \ge 2te\sqrt{1-\alpha} = \frac{e\sqrt{1-\alpha}\ln \inparen{k(1-\alpha)}}{\norm B_{\ell_1\to\ell_1}} =  \frac{e\sqrt{1-\alpha}\ln \inparen{k(1-\alpha)}}{\norm {A^\top A}_{\ell_1\to\ell_1}}\eperiod
		\end{equation*}
            \cref{eq:MainNegativeExact} follows due to the assumption that $\norm{A}_F = \sqrt n$.
		
		We turn to proving \cref{property:x_2Large,property:x_1Small,property:Ax_2Small}. For $1\le j\le n$, denote $a_j = \ip{e_{i_{\ker}}, v_j}$. Observe that $x = \sum_{j=1}^n e^{-t \lambda_j}a_j v_j$. Hence,
		\begin{align*}
			\norm{x}_2^2 &= \sum_{j=1}^n e^{-2 t \lambda_j}a_j^2 \ge \sum_{j=r+1}^n e^{-2 t \lambda_j}a_j^2 = \sum_{j=r+1}^n a_j^2 = \norm{\Pi e_{i_{\ker}}}_2^2 \ge  1-\alpha\ecomma 
		\end{align*}
		proving \cref{property:x_2Large}.
		
		For \cref{property:x_1Small}, we have
		$$\norm x_1 = \norm {e^{tB}e_{i_{\ker}}}_{1} \le  \norm {e^{tB}}_{\ell_1\to \ell_1} \norm{e_{i_{\ker}}}_{1} = \norm{e^{tB}}_{\ell_1\to \ell_1} \le \sum_{j=0}^\infty \frac{t^j\norm{B}_{\ell_1\to \ell_1}^j}{j!} = e^{t\norm{B}_{\ell_1\to \ell_1}}\eperiod$$
		
		Finally, for \cref{property:Ax_2Small} we use the inequality
		\begin{equation} \label{eq:lambdaEToTheMinusLambdaInequality}
			\lambda \cdot e^{-2t\lambda}\le \frac{1}{2te}
		\end{equation}
		for all $\lambda \ge 0$, which is readily verified via derivation of the left-hand side by $\lambda$. 
		\cref{eq:lambdaEToTheMinusLambdaInequality} yields \cref{property:Ax_2Small} since
		$$\norm{Ax}_2^2 = x^\top Bx = \sum_{j=1}^n \lambda_j a_j^2 e^{-2t\lambda_j} \le \frac{1}{2te}\sum_{j=1}^n a_j^2 = \frac{1}{2te}\eperiod $$ 
  This finishes the proof.
	\end{proof}

    Our next lemma, \cref{lem:MainNegativeARIP} below, strengthens \cref{lem:MainNegativeARIPWithoutDistortion} by replacing $\norm{B}_{\ell_1 \to \ell_1}$ with a sharper quantity. Concretely, it could be the case that $\norm{B}_{\ell_1 \to \ell_1}$ is large because a small number of columns of $B$ have large $\ell_1$-norm. By removing these columns from $A$, we can obtain a submatrix $A'$ of $A$ that is still ARIP but has $\norm{B'}_{\ell_1 \to \ell_1}$ smaller than$\norm{B}_{\ell_1 \to \ell_1}$, and this idea yields \cref{lem:MainNegativeARIP}.

\begin{lemma}\label{lem:MainNegativeARIP}
		Fix $m,n,\N$ and $s > 0$. Let $A\in \R^{m\times n}$ be a matrix of rank $\alpha n$ for some $0 < \alpha < 1$. %such that 
	    Suppose that $A$ is $(k,D)$-$\ell_2$-ARIP for some $k\ge 1$ and $D\ge \frac{1}{\sqrt{1-\alpha}}$ such that $k(1-\alpha) \ge 2$. Then, 
		\begin{equation}\label{eq:MainNegativeARIPNeedTOShow}
			\frac{\sum_{i=1}^m \norm{A_{i\Bigast }}_1^2}{\norm A_F^2}\ge \Omega\inparen{\frac{\sqrt{1-\alpha}\log\inparen{{k(1-\alpha)}}}{D^3}}\eperiod
		\end{equation}
		In particular, if $D\le O(1)$, $1-\alpha\ge \Omega(1)$ and $k \ge \Omega(n)$ then it must hold that $$\frac{\sum_{i=1}^m \norm{A_{i\Bigast }}_1^2}{\norm A_F^2}\ge \Omega(\log n)\eperiod$$
	\end{lemma}
 
	\begin{proof}		
	Assume without loss of generality that $\norm A_F = \sqrt n$. Write $B = A^\top A\in \R^{n\times n}$. Write $W = \sum_{i=1}^m \norm{A_{i\Bigast }}_1^2$  and $H = \frac{2WD^2}{n}$. Let $J = \{j\in\{1,\dots,n\}\mid \norm{B_{\Bigast j}}_1 \ge H\}$ indicate the $\ell_1$-heavy columns of $B$. Let $A'\in \R^{m\times (n-|J|)}$ be the matrix $A$ without the columns indicated by $J$. Note that the ARIP property is preserved by column-removal operations, so $A'$ is $(k,D)$-$\ell_2$-ARIP as well. Write $\alpha' = \frac{\rank(A')}{n-|J|}$.  \cref{lem:MainNegativeARIPWithoutDistortion} yields \begin{equation}\label{eq:MainNegativeARIPInitial}
	    \frac{\norm{A'^\top A'}_{\ell_1\to\ell_1}}{\norm {A'}_F^2} \ge \frac{e\sqrt{1-\alpha'}\ln\inparen{k(1-\alpha')}}{D\inparen{n-|J|}} \ge \frac{e\sqrt{1-\alpha'}\ln\inparen{k(1-\alpha')}}{Dn}\eperiod
    \end{equation}

        To deduce the lemma, we shall prove that \cref{eq:MainNegativeARIPInitial} implies \cref{eq:MainNegativeARIPNeedTOShow}. To do so, we bound some of the terms involved in \cref{eq:MainNegativeARIPInitial}. 
        
        Let $B' = A'^\top A'$ and note that $B'$ is the result of removing from $B$ the rows and columns indicated by $J$. Note that $$\norm{A'^\top A'}_{\ell_1\to \ell_1}  = \norm{B'}_{\ell_1\to \ell_1} = \max_{1\le j\le n-|J|} \norm{B'_{\Bigast j}}_1 \le \max_{j\in \inset{1,\dots, n}\setminus J}\norm{B_{\Bigast j}}_1 \le \frac{2WD^2}n\eperiod$$ We next bound $|J|$. Observe that
		\begin{align*}
			|J|&\le \frac{\sum_{j=1}^n\norm{B_{\Bigast j}}_1}{H} = \frac{\sum_{j=1}^n\sum_{i=1}^n |B_{i,j}|}H = \frac{\sum_{j=1}^n \sum_{j'=1}^n \abs{\ip{A_{\Bigast j},A_{\Bigast j'}}}}H \le \frac{\sum_{j=1}^n \sum_{j'=1}^n \sum_{i=1}^m \abs{A_{i,j}}\abs{A_{i,j'}}}H = \frac{W}H = \frac{n}{2D^2}\eperiod
		\end{align*}
		Also, by \cref{lem:BasicBounds}(\ref{enum:ColumnBound}),
		\begin{equation*}
			B_{i,i} = \norm{A_{\Bigast i}}_2^2  \le D^2
		\end{equation*}
		for all $1\le i\le n$. Hence, $$\norm{A'}_F^2 = \tr(B') \ge \tr(B) - |J|D^2 = n - |J|D^2 \ge \frac n2\eperiod$$ 
		Next, observe that $$1-\alpha' = 1-\frac{\rank B'}{n-|J|} \ge 1-\frac{\alpha n}{n-|J|} \ge 1-\frac{\alpha n}{n-\frac{n}{2D^2}} = 1- \frac{\alpha}{1-\frac 1{2D^2}} \ge 1-\frac{\alpha}{1-\frac{1-\alpha}{2}} \ge \frac{1-\alpha}2\eperiod$$
        \cref{eq:MainNegativeARIPInitial} therefore yields
        $$
        \frac{2WD^2}{n^2} \ge \frac{\frac e{\sqrt 2}\cdot \sqrt{1-\alpha}\inparen{\log\inparen{k(1-\alpha)}-\log{\sqrt 2}}}{Dn}\eperiod
        $$
        \cref{eq:MainNegativeARIPNeedTOShow} follows since
        $$\frac{\sum_{i=1}^m \norm{A_{i\Bigast }}_1^2}{\norm A_F^2} = \frac Wn \ge \Omega\inparen{\frac {\sqrt{1-\alpha}\log\inparen{k(1-\alpha)}}{D^3}}\eperiod \qedhere$$
	\end{proof}

    The final tool needed for the proof of \cref{thm:ell2ARIP} is the following \emph{row removal lemma}. Notice that ARIP is preserved by the removal of columns and addition of rows. In other words, removing rows makes it ``harder'' for a matrix to be ARIP, while removing columns makes it ``easier''. In \cref{lem:RowRemovalLemma} we show it is possible to remove any ``not too heavy'' set of rows and preserve ARIP, if one also removes a suitable set of columns.
	\begin{lemma}[Row removal lemma]\label{lem:RowRemovalLemma}
		Let $A\in \R^{m\times n}$ be $(k,D)$-$\ell_2$-ARIP. Fix $k' \le k$ and $\delta > 0$. Fix a set $I\subseteq \{1,\dots, m\}$, and let $A_I$ denote the restriction of $A$ to the row set $I$. Then, there exists a column set $J\subseteq \{1,\dots, n\}$ with $$|J| \le \frac{nD^2\norm {A_I}_F^2}{\delta^2\norm {A}_F^2}$$ such that the matrix $A_{I,J}$, defined as the matrix $A$ without the row set $I$ and the column set $J$, is $\inparen{k',\frac{D}{\sqrt{1-k'\delta^2}}}$-$\ell_2$-ARIP.
	\end{lemma}
	\begin{proof}
		Assume without loss of generality that $\norm A_F = \sqrt n$. Fix $K > 0$ such that 
		$$K\norm{x}_2 \leq \norm{Ax}_2 \leq D \cdot K\norm{x}_2$$
		for all $x\in \R^n$ with $\Delta_{1,2}(x) \ge \inparen{\frac nk}^{1/2}$. By \cref{lem:BasicBounds}(\ref{enum:OneSideARIP}), we have $\frac 1D\le K\le 1$.
		
		Let $J\subseteq \{1,\dots,n\}$ indicate the columns in $A_I$ whose $\ell_2$-norm is larger than $K\delta$. Observe that $$|J|\le \frac{\norm {A_I}_F^2}{\delta^2K^2} = \frac{n\norm {A_I}_F^2}{K^2\delta^2\norm {A}_F^2}\le \frac{nD^2\norm {A_I}_F^2}{\delta^2\norm {A}_F^2}\eperiod$$ 
		
		Define $B$ to be the matrix $A$ without the row set $I$ and the column set $J$. We need to show that $B$ is $\inparen{k',\frac{D}{\sqrt{1-k'\delta^2}}}$-$\ell_2$-ARIP. Let $A_{\overline I}$ denote the restriction of $A$ to the rows indicated by $\{1,\dots,n \}\setminus I$. Let $x\in \R^n$ have $\norm x_2=1$ and $\norm x_1 \le k'$, so that $\Delta_{1,2}(x) \ge \inparen{\frac n{k'}}^{1/2}$. Further assume that $x_j = 0$ for all $j\in J$. Note that it suffices to show that
		\begin{equation}\label{eq:RowRemovalNeedToShow}
			K\sqrt{1-k'\delta^2}\le \norm{A_{\overline I} x}_2 \le KD\eperiod
		\end{equation}
		The right-hand inequality holds since $\Delta_{1,2}(x)\ge \inparen{\frac n{k'}}^{1/2}\ge \inparen{\frac n k}^{1/2}$, implying that $\norm{A_{\overline I} x}_2 \le \norm{Ax}_2\le KD$. For the left-hand inequality of \cref{eq:RowRemovalNeedToShow}, we first note that
		\begin{equation} \label{eq:A_barIViaA_I}
		    \norm{A_{\overline I}x}_2^2 = \norm{Ax}_2^2 - \norm{A_{ I}x}_2^2\ge K^2 - \norm{A_{ I}x}_2^2 \eperiod
		\end{equation}
		Now, let $c_1,\dots, c_n$ denote the columns of $A_{I}$, and recall that $\norm {c_j}_2\le  K\delta$ for all $j\in \{1,\dots, n\}\setminus J$. Consequently,
		\begin{align*}\norm{A_{I}x}_2 &= \norm{\sum_{j=1}^n x_j c_j}_2 \le \sum_{j=1}^n \inabs{x_j} \norm{c_j}_2= \sum_{j\in\inset{1,\dots,n}\setminus J} \inabs{x_j} \norm{c_j}_2\le \norm x_1\cdot \max_{j\in\inset{1,\dots,n}\setminus J}\inset{\norm{c_j}_2} \le \norm x_1\cdot K\delta \\&\le \sqrt {k'}\cdot K\delta \eperiod\end{align*}
The left-hand inequality of \cref{eq:RowRemovalNeedToShow} now follows from the above and \cref{eq:A_barIViaA_I}.
	\end{proof}

        \subsection{Proof of \cref{thm:ell2ARIP}}
        We finally turn to proving \cref{thm:ell2ARIP} using \cref{lem:MainNegativeARIP,lem:RowRemovalLemma}.
        \begin{proof}[Proof of {\cref{thm:ell2ARIP}}]
		Suppose, without loss of generality, that $\norm A_F = \sqrt n$ and that the rows $A_{1\Bigast},\dots, A_{m\Bigast}$ are sorted so that $\Delta_{1,2}(A_{i\Bigast})$ is non-decreasing in $i$. Let $1\le t\le m$  be the minimal integer for which $\sum_{i=1}^t \norm{A_{i\Bigast}}_2^2 \ge n^{1-\eta}$. We take $T$ to be the matrix whose rows are $A_{1\Bigast},\dotsc, A_{t\Bigast}$. By definition, $\frac{\norm T_F^2}{\norm A_F^2} \ge n^{-\eta}$. By \cref{lem:BasicBounds}(\ref{enum:RowBound}), 
		$$n^{1-\eta}\le \sum_{i=1}^t \norm{A_{i\Bigast}}_2^2 \le \frac{tD^2n}k\ecomma$$
		implying that $t \ge \frac k{D^2n^\eta}$. This proves that $T$ satisfies \cref{property:TLargeEll2ARIP,property:TmanyRowsARIP}. 
		
		To prove \cref{property:TDistortedARIP}, it suffices to show that \begin{equation}\label{eq:ell2NeedToShow}
		    \Delta_{1,2}\inparen{A_{(t+1)\Bigast}} \le O\inparen{\inparen{\frac{D^3 n}{\eta \sqrt{1-\alpha}\log n}}^{1/2}}\eperiod
		\end{equation} Indeed, \cref{eq:ell2NeedToShow} implies \cref{eq-arip-distortion} since $\Delta_{1,2}(A_{i\Bigast})$ is non-decreasing in $i$.

          Let $A'\in \R^{(m-t)\times n}$ be the matrix whose rows are $A_{(t+1)\Bigast},\dots,A_{m\Bigast}$. We apply \cref{lem:RowRemovalLemma} to the matrix $A$, with $I = \inset{1,\dotsc, t}$, $k' = \frac{(1-\alpha)n^\eta}{8D^4}$ and $\delta = \sqrt{\frac {1}{2k'}}$. The lemma yields a $\inparen{k', \sqrt 2\cdot  D}$-$\ell_2$-ARIP submatrix $S\in \R^{(m-t)\times (n-w)}$ of $A'$, where $$w \le \frac{nD^2\norm T_F^2}{\delta^2\norm A_F^2} = \frac{D^2\norm T_F^2}{\delta ^2} = 2D^2\norm{T}_F^2k'\eperiod$$
		By the minimality of $t$ and \cref{lem:BasicBounds}(\ref{enum:RowBound}), 
		$$\norm T_F^2 \le n^{1-\eta} + \norm {A_{t\Bigast}}_2^2 \le n^{1-\eta} + \frac{D^2n}k\le 2 n^{1-\eta}\ecomma$$
  where the last step uses the hypothesis $k \ge n^\eta D^2$. We therefore have 
		$$w\le  4 D^2n^{1-\eta} k' = \frac{(1-\alpha)n}{2D^2}\eperiod$$ 
		
		Let $$\alpha' = \frac {\rank(S)}{n-w} \le \frac{\rank(A)}{n-w} = \frac{\alpha n}{n-w}$$ and note that
		$$1-\alpha' \ge 1-\frac{\alpha n}{n-w} \ge 1-\frac{\alpha}{1-\frac{1-\alpha}{2D^2}}\ge 1-\frac{\alpha}{1-\frac{1-\alpha}2}\ge \frac{1-\alpha}2\eperiod$$ Therefore, \cref{lem:MainNegativeARIP} yields
		\begin{equation}\label{eq:SIsDense}		    \frac{\sum_{i=1}^{m-t}\norm{S_{i\Bigast}}_1^2}{\norm S_F^2} \ge \Omega\inparen{\frac{\sqrt{1-\alpha'}\log\inparen{k'(1-\alpha')}}{D^3}} \ge \Omega\inparen{\frac{\sqrt{1-\alpha}\log \inparen{\frac{(1-\alpha)^2n^\eta}{8D^4}}}{D^3}}\ge \Omega\inparen{\frac{\eta \sqrt{1-\alpha}\log n}{D^3}}\eperiod
		\end{equation}
		
		Let $c_1,\dots, c_w\in \R^m$ denote the columns of $A$ that are missing from $S$. By \cref{lem:BasicBounds}(\ref{enum:ColumnBound}), $\norm {c_j}_2 \le D$ for all $1\le j\le w$. Hence,
		$$\norm S_F^2 \ge \norm {A}_F^2 - \norm T_F^2 -\sum_{j=1}^w \norm {c_j}_2^2 \ge n - 2n^{1-\eta} - wD^2\ge n - 2n^{1-\eta} - \frac n2 \ge \frac n4\eperiod$$ 
		Consequently,		
		\begin{align*}\frac{\sum_{i=1}^{m-t} \norm{S_{i\Bigast}}_1^2}{\norm S_F^2} &\le \frac{4 \sum_{i=1}^{m-t}\norm {S_{i\Bigast}}_1^2}{n} \le \frac{4\sum_{i=t+1}^m \norm{A_{i\Bigast}}_1^2}n = \sum_{i={t+1}}^m \frac{4\norm{A_{i\Bigast}}_2^2}{\Delta_{1,2}\inparen{A_{i\Bigast}}^2} \le \sum_{i=t+1}^m \frac{4\norm {A_{i\Bigast}}_2^2}{\Delta_{1,2}(A_{{(t+1)}\Bigast})^2}\\&\le \frac {4\norm A_F^2}{\Delta_{1,2}(A_{(t+1)\Bigast})^2}= \frac{4n}{\Delta_{1,2}\inparen{A_{(t+1)\Bigast}}^2}\eperiod\end{align*}
        \cref{eq:ell2NeedToShow}, which yields \cref{property:TDistortedARIP}, follows from the above and \cref{eq:SIsDense}.
	\end{proof}

	%%%%%%%%%%%%%%%%%%%%%%%%%%%%%%%%%%%%%%%%%%%%%%%%%%%%%%%%%%%%%%%%%%%%%%%%%%%%%%%%
	%%%%%%%%%%%%%%%%%%%%%%%%%%%%%%%%%%%%%%%%%%%%%%%%%%%%%%%%%%%%%%%%%%%%%%%%%%%%%%%%
	%%%%%%%%%%%%%%%%%%%%%%%%%%%%%%%%%%%%%%%%%%%%%%%%%%%%%%%%%%%%%%%%%%%%%%%%%%%%%%%%
	\section{For $p \ne 2$, $\ell_p$-RIP Matrices Must Be Sparse}
	In this section, we prove \cref{ithm:ellp}, which we formally state below.
     \begin{mtheorem}\label{thm:ellp}
    		Let $A \in \R^{m \times n}$ be a $(k, D)$-$\ell_p$-RIP matrix, and let $A_{1 \Bigast}, \dots, A_{m \Bigast}$ denote the rows of $A$. Then, if $1 \leq p < 2$, it holds that
    		\begin{flalign*}
    			D^p \left(\frac{n}{k}\right)^{p(\frac{1}{p} - \frac{1}{2})} \sum_{i = 1}^m  \norm{A_{i \Bigast}}_2^{p} \geq \sum_{i = 1}^m \norm{A_{i \Bigast}}_p^p \enspace,
    		\end{flalign*}
    		and if $p > 2$ it holds that 
    		\begin{flalign*}
    			\sum_{i = 1}^m  \norm{A_{i \Bigast}}_2^{p} \leq D^p  \left(\frac{n}{k}\right)^{p(\frac{1}{2} - \frac{1}{p} )}\sum_{i = 1}^m \norm{A_{i \Bigast}}_p^p \enspace.
    		\end{flalign*}
     \end{mtheorem}

	We will need the following simple claim.
	\begin{claim}
		\label{claim:simple2}
		Suppose $X, Y$ are non-negative random variables, and $\Pr[Y = 0] = 0$. Then, $\Pr[ X/Y \leq \E[X]/\E[Y]] > 0$ and $\Pr[ X/Y \geq \E[X]/\E[Y]] > 0$.
	\end{claim}
	\begin{proof}
		Let $\alpha = \E[X]/\E[Y]$. Then, $\E[X - \alpha Y] = 0$. So, $\Pr[X - \alpha Y \leq 0] > 0$. Therefore, $\Pr[X - \alpha Y \leq 0 \wedge Y > 0] > 0$, and so $\Pr[X/Y \leq \alpha] > 0$. Similarly, $\Pr[X - \alpha Y \geq 0] > 0$, and so $\Pr[X - \alpha Y \geq 0 \wedge Y > 0] > 0$, which implies $\Pr[X/Y \geq \alpha] > 0$. This finishes the proof.
	\end{proof}
	
	\begin{proof}[Proof of \cref{thm:ellp}]
		For $j \in [n]$, let $e_j$ denote the $j$-th standard basis vector. Without loss of generality, we shall assume that $A$ satisfies \cref{def:RIP} with $K = 1$; otherwise, we can rescale $A$ so that this holds.
		Let $A_{\Bigast 1}, \dots, A_{\Bigast n}$ be the columns of $A$. Observe that $\norm{A e_j}_p = \norm{A_{\Bigast j}}_p$ for all $j \in [n]$. As $k \geq 1$, we thus have that $1 \leq \norm{A_{\Bigast j}}_p \leq D$, for all $j \in [n]$. It thus follows that $n \leq \sum_{i = 1}^m \norm{A_{i \Bigast}}_p^p = \sum_{j = 1}^n \norm{A_{\Bigast j}}_p^p \leq n D^p$.
		
		Now, let $S \subseteq [n]$, $\abs{S} \leq k$. For $j \in S$, let $g_j \sim N(0, 1)$, and let $x =\sum_{j \in S} g_j e_j$. Note that $x \in \R^S$. We observe that $\norm{Ax}_p^p$ and $\norm{x}_p^p$ are nonnegative random variables, and $\Pr[\norm{x}_p^p = 0] = 0$.
		
		Next, we note that if $g \sim N(0, 1)$, then $\E[\abs{g}^p] = \frac{2^{p/2}}{\sqrt{\pi}} \cdot \Gamma(\frac{1 + p}{2}) =: f(p)$. By linearity of expectation, it then follows that $\E[\norm{x}_p^p] = f(p) \abs{S}$, and that $\E[\norm{Ax}_p^p] = f(p) \sum_{i = 1}^m \norm{A_{i, S}}_2^p$, where $A_{i, S}$ denotes the $i$-th row restricted to the coordinates in $S$.
		
		We now have two cases.
		
		\parhead{Case 1: $p < 2$.}
		Applying \cref{claim:simple2}, we see that there exists $y \in \R^S \setminus \{0^S\}$ such that $\norm{Ay}_p^p/\norm{y}_p^p \leq \E[\norm{Ax}_p^p]/\E[\norm{x}_p^p]$.
		
		It follows that there exists $y \in \R^S \setminus \{0^S\}$ with $\norm{y}_p = 1$ such that $\norm{Ay}_p^p \leq \sum_{i = 1}^m \norm{A_{i, S}}_2^p/\abs{S}$.
		On the other hand, because $y$ is $k$-sparse, we see that $\norm{Ay}_p^p \geq 1$. Hence, 
		\begin{equation}
			\sum_{i = 1}^m \norm{A_{i, S}}_2^p \geq \abs{S} \label{eq:main} \enspace,
		\end{equation}
		for every $S$ of size at most $k$.
		
		Now, fix $i$, and let $X$ denote the random variable $\norm{A_{i, S}}_2^p$, with randomness over the draw of $S \subseteq [n]$, $\abs{S} = k$. By H\"{o}lder's inequality (and using that $p < 2$), we have $\E[X] \leq \E[X^{2/p}]^{p/2}$. Now, $\E[X^{2/p}] = \E_S[\norm{A_{i, S}}_2^2] = \frac{k}{n} \norm{A_{i \Bigast}}_2^2$. This is because each coordinate of $n$ appears in a randomly chosen $S$ with probability $\frac{k}{n}$. It thus follows that $\E[X] \leq (k/n)^{p/2} \norm{A_{i \Bigast}}_2^{p}$.
		
		Taking expectations of \cref{eq:main} over the choice of $\abs{S} = k$, we now have that 
		\begin{flalign*}
			(k/n)^{p/2} \sum_{i = 1}^m  \norm{A_{i \Bigast}}_2^{p} \geq k \enspace.
		\end{flalign*}
		Combining with the inequality $\sum_{i = 1}^m \norm{A_{i \Bigast}}_p^p  \leq n D^p$, we thus have
		\begin{flalign*}
			&(k/n)^{p/2} \sum_{i = 1}^m  \norm{A_{i \Bigast}}_2^{p} \geq k \geq \frac{k}{n} D^{-p} \sum_{i = 1}^m \norm{A_{i \Bigast}}_p^p \\
			&\implies D^p \left(\frac{n}{k}\right)^{p(\frac{1}{p} - \frac{1}{2})} \sum_{i = 1}^m  \norm{A_{i \Bigast}}_2^{p} \geq \sum_{i = 1}^m \norm{A_{i \Bigast}}_p^p \enspace,
		\end{flalign*}
		as required.
		
		\parhead{Case 2: $p > 2$.} Applying \cref{claim:simple2}, we see that exists $y \in \R^S \setminus \{0^S\}$ such that $\norm{Ay}_p^p/\norm{y}_p^p \geq \E[\norm{Ax}_p^p]/\E[\norm{x}_p^p]$.
		
		It follows that there exists $y \in \R^S \setminus \{0^S\}$ with $\norm{y}_p = 1$ such that $\norm{Ay}_p^p \geq \sum_{i = 1}^m \norm{A_{i, S}}_2^p/\abs{S}$.
		On the other hand, because $y$ is $k$-sparse, we see that $\norm{Ay}_p^p \leq D^p$. Hence, 
		\begin{equation}
			\sum_{i = 1}^m \norm{A_{i, S}}_2^p \leq \abs{S} D^p \label{eq:main2} \enspace,
		\end{equation}
		for every $S$ of size at most $k$.
		
		Now, fix $i$, and let $X$ denote the random variable $\norm{A_{i, S}}_2^p$, with randomness over the draw of $S \subseteq [n]$, $\abs{S} = k$. By H\"{o}lder's inequality (and using that $p > 2$), we have $\E[X] \geq \E[X^{2/p}]^{p/2}$. Now, $\E[X^{2/p}] = \E_S[\norm{A_{i, S}}_2^2] = \frac{k}{n} \norm{A_{i \Bigast}}_2^2$. This is because each coordinate of $n$ appears in a randomly chosen $S$ with probability $\frac{k}{n}$. It thus follows that $\E[X] \geq (k/n)^{p/2} \norm{A_{i \Bigast}}_2^{p}$.
		
		Taking expectations of \cref{eq:main2} over the choice of $\abs{S} = k$, we now have that 
		\begin{flalign*}
			(k/n)^{p/2} \sum_{i = 1}^m  \norm{A_{i \Bigast}}_2^{p} \leq k D^p \enspace.
		\end{flalign*}
		Combining with the inequality $\sum_{i = 1}^m \norm{A_{i \Bigast}}_p^p  \geq n$, we thus have
		\begin{flalign*}
			&(k/n)^{p/2} \sum_{i = 1}^m  \norm{A_{i \Bigast}}_2^{p} \leq k D^p \leq \frac{k}{n} D^p \sum_{i = 1}^m \norm{A_{i \Bigast}}_p^p \\
			&\implies  \sum_{i = 1}^m  \norm{A_{i \Bigast}}_2^{p} \leq D^p  \left(\frac{n}{k}\right)^{p(\frac{1}{2} - \frac{1}{p} )}\sum_{i = 1}^m \norm{A_{i \Bigast}}_p^p \enspace,
		\end{flalign*}
		as required.
	\end{proof}

	%%%%%%%%%%%%%%%%%%%%%%%%%%%%%%%%%%%%%%%%%%%%%%%%%%%%%%%%%%%%%%%%%%%%%%%%%%%%%%%%
	%%%%%%%%%%%%%%%%%%%%%%%%%%%%%%%%%%%%%%%%%%%%%%%%%%%%%%%%%%%%%%%%%%%%%%%%%%%%%%%%
	%%%%%%%%%%%%%%%%%%%%%%%%%%%%%%%%%%%%%%%%%%%%%%%%%%%%%%%%%%%%%%%%%%%%%%%%%%%%%%%%
	\printbibliography
	%%%%%%%%%%%%%%%%%%%%%%%%%%%%%%%%%%%%%%%%%%%%%%%%%%%%%%%%%%%%%%%%%%%%%%%%%%%%%%%%
	%%%%%%%%%%%%%%%%%%%%%%%%%%%%%%%%%%%%%%%%%%%%%%%%%%%%%%%%%%%%%%%%%%%%%%%%%%%%%%%%
	%%%%%%%%%%%%%%%%%%%%%%%%%%%%%%%%%%%%%%%%%%%%%%%%%%%%%%%%%%%%%%%%%%%%%%%%%%%%%%%%
\end{document}